\documentclass[11pt]{article}

\usepackage[a4paper,margin=1in]{geometry}
\usepackage{amsmath,amssymb,amsthm,mathtools}
\usepackage{microtype}
\usepackage[colorlinks=true,linkcolor=blue,citecolor=blue,urlcolor=blue]{hyperref}
\usepackage{hyperref}
\usepackage{cite}
\newtheorem{theorem}{Theorem}
\newtheorem{lemma}{Lemma}
\newtheorem{proposition}{Proposition}
\newtheorem{corollary}{Corollary}
\newtheorem{remark}{Remark}

\DeclareMathOperator{\rank}{rank}

\title{A Slice-Rank Drift Bound for Random Quantum \(k\)-SAT}
\author{Jean Bernoulli Ravelomanana}
\date{}

\begin{document}

\maketitle

\begin{abstract}
	Random quantum satisfiability is a natural quantum analogue of random
	constraint satisfaction and a basic model for frustration-free local
	Hamiltonians.  Despite extensive work on its satisfiable and unsatisfiable
	regimes, the quantitative location of the random quantum \(k\)-SAT threshold
	has remained poorly understood, with the best general upper bounds leaving a
	large gap to the known lower bounds.  In this paper we prove a new upper bound
	on the satisfiability threshold of random quantum \(k\)-SAT.  Our result
	improves the previously known asymptotic upper bound by a factor of order
	\(k\), giving a bound of order \(2^k/k\).  The improvement is also significant
	at small values of \(k\); in particular, for random quantum \(3\)-SAT we obtain
	a substantially smaller explicit upper bound than the one previously
	available.  The proof combines the geometric formulation of generic quantum
	satisfiability with a dimension-decay analysis of the full satisfying
	subspace.  The key input is a multiplicative Shearer-type inequality for
	tensor-product subspaces, which quantifies how global dimension forces
	nontrivial local dimension on typical sets of qubits.
\end{abstract}

\section{Introduction}\label{sec:introduction}
\subsection{From random \(k\)-SAT to random quantum \(k\)-SAT}
\label{subsec:motivation-history}

Random constraint satisfaction problems form a fruitful area of research at the
intersection of probability, combinatorics, theoretical computer science and
statistical physics.  The paradigmatic example is random \(k\)-SAT, in which
one asks whether a random Boolean formula built from clauses involving \(k\)
variables admits a satisfying assignment.  

The passage from worst-case
complexity to random models is central here: worst-case
\(\mathrm{NP}\)-completeness shows that \(k\)-SAT is computationally hard in
general, but it does not describe the behaviour of typical formulas, nor how
satisfiability depends on the density of constraints.  Random \(k\)-SAT
addresses this question by using the clause density as a control parameter.
As this density increases, typical instances pass from satisfiable to
unsatisfiable, and the set of solutions changes abruptly from exponentially
large to empty.  This SAT--UNSAT transition has motivated a long line of
rigorous and non-rigorous work on thresholds, the geometry of the solution
space and algorithmic barriers
\cite{MezardMontanari09,Friedgut99,AchlioptasPeres04,DingSlySun22,AchlioptasMoore06,MezardParisiZecchina02,MertensMezardZecchina06, KrzakalaMontanariRicciSemerjianZdeborova07,AchlioptasCojaOghlan08,KirkpatrickSelman94,MonassonZecchina99}.

Quantum \(k\)-SAT is the natural quantum analogue of this picture.  It was
introduced by Bravyi~\cite{Bravyi06} as a quantum analogue of classical
\(k\)-SAT and as a canonical satisfiability problem for frustration-free local
Hamiltonians.  In quantum \(k\)-SAT, clauses are replaced by local projectors
acting on \(k\) qubits, and satisfying assignments are replaced by common
zero-energy states.  Thus a quantum \(k\)-SAT instance is satisfiable precisely
when the corresponding local Hamiltonian is frustration-free.  From the
viewpoint of computational complexity, Bravyi proved that quantum \(2\)-SAT is
solvable in polynomial time, whereas quantum \(k\)-SAT is
\(\mathrm{QMA}_1\)-complete for \(k\geq 4\)~\cite{Bravyi06}.  The corresponding
hardness result for \(k=3\) was later proved by Gosset and
Nagaj~\cite{GossetNagaj16}.  Here \(\mathrm{QMA}_1\) is the quantum analogue of
\(\mathrm{NP}\).\footnote{A yes-instance of a problem in \(\mathrm{QMA}\)
	admits a polynomial-size quantum witness that can be checked by a
	polynomial-time quantum verifier.  The subscript in \(\mathrm{QMA}_1\)
	indicates one-sided completeness, meaning that yes-instances are accepted with
	probability one.}

The random quantum model is motivated by the analogous passage from worst-case
complexity to typical behaviour.  Once quantum \(k\)-SAT has been identified
as a quantum analogue of classical \(k\)-SAT, it is natural to ask what happens
for typical quantum instances.  Random quantum \(k\)-SAT provides such an
ensemble: the density of local projectors plays the role of the clause
density, and the central question becomes the typical transition between
frustration-free and frustrated quantum systems.  This makes random quantum
\(k\)-SAT one of the simplest random models in which constraint-satisfaction
phase transitions, local Hamiltonian complexity and the geometry of many-body
Hilbert space meet.

The random quantum \(k\)-SAT model, including its generic formulation, was
introduced by Laumann et al.~\cite{LMS10,LLMSS10}.  In the random generic
model, one first chooses a random \(k\)-uniform interaction hypergraph and then
assigns a generic local projector to each hyperedge.  A crucial feature of this
model is that, once the hypergraph is fixed, the dimension of the satisfying
subspace is almost surely independent of the particular choice of generic
projectors.  The random quantum \(k\)-SAT threshold is therefore controlled by
how local constraints reduce the dimension of a tensor-product subspace.

The quantum problem is not merely a formal variant of classical \(k\)-SAT.  A
classical instance asks for a point in a discrete cube satisfying a collection
of local Boolean constraints.  A quantum instance asks instead for a nonzero
vector in the intersection of many local kernels inside the tensor product
space \((\mathbb C^2)^{\otimes n}\).  The constraints are local, but the object
whose dimension must be controlled is global and linear-algebraic.  This
tension between locality, randomness and tensor-product geometry is the main
source of difficulty in obtaining quantitative bounds for random quantum
satisfiability.  The next question is therefore what rigorous bounds are known
for this transition, and how far they are from locating the threshold.

\subsection{Known bounds and the remaining gap}
\label{subsec:known-bounds}

The satisfiable side of the transition is governed by quantum analogues of the
Lov\'asz local lemma.  Ambainis, Kempe and Sattath proved a quantum Lov\'asz
local lemma and applied it to random quantum \(k\)-SAT, obtaining a satisfiable
bound of order \(2^k/k^2\)~\cite{AKS}.  This line of work was later refined
through Shearer-type bounds and their connection with hard-core lattice gases
\cite{SMLM16}.  The tightness of Shearer's bound for the quantum Lov\'asz local
lemma was subsequently established in~\cite{HLSZ19}.  These results give
rigorous lower bounds on the satisfiable region and show that random quantum
\(k\)-SAT remains satisfiable at densities growing exponentially with \(k\),
although polynomially below \(2^k\).

A more refined question, already emphasized in the early work on the random
generic model~\cite{LLMSS10}, is whether the satisfying state can be chosen
without entanglement.  In other words, one asks whether there is a satisfying
state obtained by assigning an individual one-qubit state to each vertex and
taking their tensor product.  Such an instance is called product satisfiable.
Product satisfiability is stronger than satisfiability, but it captures an
important part of the geometry of the random model and provides a bridge
between quantum constraints and classical combinatorial structures.  In the
random generic model, the disappearance of product satisfying states is
governed by a matching, or dimer-covering, condition on the underlying
interaction hypergraph.  In recent work by Lee et al.~\cite{LMRV24}, the
product-satisfiable phase was revisited and an alternative proof of the
PRODSAT threshold was given, based on this dimer-covering formulation.

The product-satisfiability threshold is expected to occur before the full
satisfiability threshold.  Consequently, above the PRODSAT threshold but below
the quantum \(k\)-SAT threshold, one is led to a genuinely quantum regime: if
satisfying states still exist, then they cannot be product states and must be
entangled.  This possible intermediate ENTSAT phase is one of the distinctive
features of random quantum satisfiability.  Numerical work for random
\(3\)-QSAT suggests that the situation at small \(k\) may be more delicate,
and that the product-satisfiability and full satisfiability transitions may
possibly coincide in that case~\cite{MorampudiHSLM17}.
Nevertheless, the distinction between product satisfiability and full quantum
satisfiability explains why controlling the full satisfying subspace is
substantially more difficult than controlling product solutions alone.

On the unsatisfiable side, the strongest general estimates before the present
work were obtained by Bravyi, Moore and Russell~\cite{BMR}.  Their method gives
an upper bound of order \(2^k\) on the random quantum \(k\)-SAT threshold.  In
the first nontrivial case \(k=3\), they proved the explicit bound
\[
\alpha^{\mathrm{QSAT}}_3 \leq 3.594 .
\]
Thus the rigorous asymptotic picture contained a large window: random quantum
\(k\)-SAT was known to be satisfiable up to densities of order \(2^k/k^2\),
while the best general unsatisfiable bound was only of order \(2^k\).  Between
these two scales lies the difficult question of how long the full satisfying
subspace survives after product satisfying states have disappeared.

The present paper addresses this UNSAT-side problem directly.  Rather than
restricting attention to product states or to special local obstructions, we
study the decay of the full satisfying subspace under the addition of random
generic constraints.  This leads to a new general upper bound of order
\(2^k/k\) for the random quantum \(k\)-SAT threshold, improving the previous
asymptotic upper bound by a factor of order \(k\).  In the first nontrivial
case \(k=3\), the method gives a significantly better explicit upper bound
than the previous value \(3.594\) from~\cite{BMR}.

\subsection{Organization of the paper}\label{subsec:organization}

The rest of the paper is organized as follows.  Section~\ref{sec:model-main-strategy}
defines the random generic quantum \(k\)-SAT model, states the main
theorem, and gives the proof strategy.  Section~\ref{sec:contraction-ranks}
identifies the dimension lost when a generic constraint is imposed with a
contraction rank.  Section~\ref{sec:slice-rank} proves the multiplicative
slice-rank inequality and derives the local-rank estimates for random
\(k\)-sets.  Section~\ref{sec:dimension-drift-proof} turns these estimates into
a logarithmic dimension drift, integrates the drift by a stopped martingale
argument, transfers the result to the model
with distinct supports, and completes the proof of Theorem~\ref{thm:main}.

\section{Model, main result and proof strategy}\label{sec:model-main-strategy}

\subsection{Definition of the random quantum \(k\)-SAT model}
\label{subsec:model}

We now fix the model used in the main theorem.  We introduce it by starting
from the classical \(k\)-SAT problem and then passing to its quantum analogue.
In classical \(k\)-SAT, one considers Boolean variables
\(x_1,\ldots,x_n\in\{0,1\}\) and a formula
\[
F = \bigwedge_{a=1}^m C_a
\]
in conjunctive normal form.  Each clause \(C_a\) is a disjunction of \(k\)
literals involving \(k\) of the variables.  Such a clause forbids exactly one
of the \(2^k\) possible local assignments of these variables.  Equivalently,
one may associate with \(F\) the non-negative classical Hamiltonian
\[
h_F(x_1,\ldots,x_n)
=
\sum_{a=1}^m
\mathbf 1_{\{C_a \text{ is violated by } x\}} .
\]
The formula is satisfiable if and only if this Hamiltonian has a zero-energy
configuration:
\[
\min_{x\in\{0,1\}^n} h_F(x)=0 .
\]

Quantum \(k\)-SAT is obtained by replacing Boolean variables by qubits and
forbidden local assignments by forbidden local quantum states.  The state space
of \(n\) qubits is
\[
\mathcal H_n = (\mathbb C^2)^{\otimes n}.
\]
For a set \(T\subset[n]\), we write
\[
\mathcal H_T = (\mathbb C^2)^{\otimes T},
\qquad
\bar T=[n]\setminus T,
\]
so that \(\mathcal H_n=\mathcal H_T\otimes\mathcal H_{\bar T}\).

In this paper, a quantum \(k\)-SAT constraint supported on a set
\(T\subset[n]\), with \(|T|=k\), is a rank-one constraint specified by a
nonzero vector \(\xi_T\in\mathcal H_T\).  It forbids the one-dimensional
local direction spanned by \(\xi_T\) on the \(k\) qubits in \(T\), and acts
trivially on all other qubits.  Higher-rank local projectors can also be
considered, but throughout this paper quantum \(k\)-SAT refers to the
standard rank-one model.  After normalizing \(\xi_T\), the corresponding
projector is
\[
\Pi_T = |\xi_T\rangle\langle \xi_T|\otimes I_{\bar T}.
\]
Thus a collection of constraints \(T_1,\ldots,T_m\), together with local
vectors \(\xi_{T_a}\in\mathcal H_{T_a}\), defines the local Hamiltonian
\[
H = \sum_{a=1}^m \Pi_{T_a}.
\]
The instance is satisfiable if there exists a nonzero state
\(\psi\in\mathcal H_n\) annihilated by every local projector:
\[
\Pi_{T_a}\psi=0,
\qquad a=1,\ldots,m.
\]
Equivalently, if
\[
\mathcal S
=
\bigcap_{a=1}^m \ker \Pi_{T_a},
\]
then the instance is satisfiable precisely when
\[
\mathcal S\neq\{0\}.
\]
Since the projectors are positive semidefinite, this is the same as saying that
the Hamiltonian \(H\) has a nonzero zero-energy ground space:
\[
\mathcal S=\ker H .
\]
Thus quantum \(k\)-SAT is the frustration-free satisfiability problem for
local Hamiltonians with rank-one constraints.

This construction extends classical \(k\)-SAT in a direct way.  A classical
clause forbids one of the \(2^k\) local Boolean assignments.  A quantum
constraint in the rank-one model forbids one direction in the
\(2^k\)-dimensional Hilbert space \(\mathcal H_T\).  When all forbidden local vectors are computational
basis vectors, the Hamiltonian above is diagonal in the computational basis and
reduces to a classical \(k\)-SAT cost function.  For generic vectors, however,
the constraints are not diagonal, and the satisfying space is a genuine linear
subspace of the full tensor-product Hilbert space.

In this paper we study the random generic model.  The supports
\(T_1,\ldots,T_m\) are sampled independently and uniformly from the
\(k\)-subsets of \([n]\).  Conditional on these supports, the local constraints
are chosen generically: equivalently, the covectors
\[
\varphi_{T_a}=\langle \xi_{T_a}|
\in \mathcal H_{T_a}^*
\]
are sampled independently from distributions that are absolutely continuous
with respect to Lebesgue measure in any system of linear coordinates.  The
constraint imposed on a global state \(\psi\in\mathcal H_n\) may then be
written as
\[
(\varphi_{T_a}\otimes I_{\bar T_a})\psi=0,
\qquad a=1,\ldots,m .
\]
The normalization of the local vectors is irrelevant, since only the kernels of
the projectors matter.

Equivalently, the random supports \(T_1,\ldots,T_m\) define a random
\(k\)-uniform interaction hypergraph on the vertex set \([n]\), whose
hyperedges indicate which qubits participate in each local constraint.  In the
with-replacement model used in the proof, this is a \(k\)-uniform
multihypergraph, since the same support may appear more than once.  The
corresponding model with distinct supports is the usual simple random
\(k\)-uniform hypergraph model.  The distinction is immaterial for the
asymptotic upper bound proved here and is addressed in the final transfer
argument.

Throughout the paper, \(k\geq 3\) is fixed and \(m=\alpha n\).  All almost-sure
statements concerning the local constraints are with respect to the generic
choice of covectors.  We use ``with high probability'' to mean with probability
tending to one as \(n\to\infty\).

\subsection{Main theorem}
\label{subsec:main-theorem}

We now state the main result.  It is useful first to introduce threshold
notation that does not presuppose the existence of a sharp threshold.  Define
\[
\alpha_{\mathrm{sat}}(k)
=
\sup\Bigl\{
\alpha\geq 0:
\mathbb P\bigl(\mathcal S\neq\{0\}\bigr)\to 1
\text{ for } m=\alpha n
\Bigr\}
\]
and
\[
\alpha_{\mathrm{unsat}}(k)
=
\inf\Bigl\{
\alpha\geq 0:
\mathbb P\bigl(\mathcal S=\{0\}\bigr)\to 1
\text{ for } m=\alpha n
\Bigr\}.
\]
Here the probabilities refer to the random generic quantum
\(k\)-SAT model of Section~\ref{subsec:model}.  Thus
\[
\alpha_{\mathrm{sat}}(k)
\leq
\alpha_{\mathrm{unsat}}(k).
\]
If the model has a sharp satisfiability threshold, then these two quantities
coincide; their common value is the usual critical density, which we denote by
$
\alpha_c(k).
$

For \(k\geq 3\), set
\[
\alpha_\star(k)
=
\int_0^1
\frac{ds}{
	\log_2\left(\frac{1}{1-2^{-ks}}\right)
} .
\]

\begin{theorem}[Random quantum \(k\)-SAT upper bound]
	\label{thm:main}
	Fix \(k\geq 3\).  Consider the random generic quantum \(k\)-SAT
	instance on \(n\) qubits obtained by sampling \(m\) independent constraints as
	in Section~\ref{subsec:model}.  If, for some \(\varepsilon>0\),
	\[
	m\geq \bigl(\alpha_\star(k)+\varepsilon\bigr)n ,
	\]
	then the instance is unsatisfiable with high probability as \(n\to\infty\).
	Equivalently,
	\[
	\alpha_{\mathrm{unsat}}(k)
	\leq
	\alpha_\star(k).
	\]
	The same conclusion holds for the corresponding model with distinct supports,
	or equivalently for the simple random \(k\)-uniform interaction hypergraph
	model.  Moreover,
	\[
	\alpha_\star(k)
	\sim
	\frac{2^k}{k}
	\qquad\text{as } k\to\infty .
	\]
\end{theorem}
Thus, if a sharp satisfiability threshold exists, Theorem~\ref{thm:main}
implies
\[
\alpha_c(k)
\leq
\alpha_\star(k)
\sim
\frac{2^k}{k}.
\]
This should be compared with the previous general upper bounds for random quantum \(k\)-SAT, which were of order \(2^k\). 
For the first nontrivial value \(k=3\), the estimate gives a substantial
numerical improvement.

\begin{corollary}[The three-local case]
	\label{cor:three-local}
	For random generic quantum \(3\)-SAT,
	\[
	\alpha_{\mathrm{unsat}}(3)
	\leq
	\alpha_\star(3)
	=
	\int_0^1
	\frac{ds}{
		\log_2\left(\frac{1}{1-2^{-3s}}\right)
	}
	\approx 1.947 .
	\]
	In particular, if a sharp threshold exists, then
	\[
	\alpha_c(3)
	\leq
	1.947 .
	\]
\end{corollary}
This improves the previous upper bound
\[
\alpha_c(3)\leq 3.594
\]
of Bravyi, Moore and Russell~\cite{BMR}.
Combining Theorem~\ref{thm:main} with the quantum Lov\'asz local lemma gives
the following updated rigorous window.

\begin{corollary}[Known asymptotic window]
	\label{cor:window}
	For random generic quantum \(k\)-SAT,
	\[
	\Omega\!\left(\frac{2^k}{k^2}\right)
	\leq
	\alpha_{\mathrm{sat}}(k)
	\leq
	\alpha_{\mathrm{unsat}}(k)
	\leq
	\alpha_\star(k)
	\sim
	\frac{2^k}{k}.
	\]
	Consequently, if a sharp satisfiability threshold exists, then its critical
	density satisfies
	\[
	\Omega\!\left(\frac{2^k}{k^2}\right)
	\leq
	\alpha_c(k)
	\leq
	\alpha_\star(k)
	\sim
	\frac{2^k}{k}.
	\]
\end{corollary}
The corollary records the resulting rigorous window for the random generic
model.  We next explain the main idea behind the proof and the sense
in which the dimension-decay method differs from the previous approaches to
UNSAT-side bounds.

\subsection{Proof strategy and novelty of the method}
\label{subsec:discussion}

We briefly explain the structure of the proof of Theorem~\ref{thm:main}.  The
argument proceeds in three steps.  First, we translate the effect of adding a
generic rank-one constraint into a linear-algebraic dimension drop.  This is
done in Section~\ref{sec:contraction-ranks} through contraction ranks.  Second,
we prove a deterministic tensor inequality, the multiplicative slice-rank
inequality, which forces a subspace of large global dimension to have large
generic contractions on typical local sets of coordinates.  This is the content
of Sections~\ref{sec:slice-rank} and~\ref{sec:shadow}.  Third, we turn these
local contraction estimates into a drift estimate for the logarithm of the
dimension of the satisfying subspace.  The drift is integrated by a stopped
martingale argument in Sections~\ref{sec:drift} and~\ref{sec:martingale},
and the final technical reductions are completed in
Section~\ref{sec:completion}.

We now describe the key mechanism.  Suppose that some constraints have already
been imposed, and let
\[
W\subseteq \mathcal H_n=(\mathbb C^2)^{\otimes n}
\]
be the current satisfying subspace.  If a new generic rank-one constraint is
placed on a set \(T\subset[n]\), then the dimension drop is governed by the
contraction map
\[
C_T(\varphi_T):W\to\mathcal H_{\bar T},
\qquad
C_T(\varphi_T)(\psi)=(\varphi_T\otimes I_{\bar T})\psi .
\]
We denote by \(r_T(W)\) the generic rank of this map.  Thus, if
\(D=\dim W\), imposing the new constraint changes the dimension from \(D\) to
\[
D-r_T(W)
\]
almost surely.  The problem is therefore to show that a fresh random local
constraint typically has a sufficiently large contraction rank on the current
satisfying subspace.

The new deterministic input is the multiplicative slice-rank inequality.  For
every nonzero subspace \(W\subseteq(\mathbb C^2)^{\otimes n}\), if \(r_i(W)\)
denotes the generic contraction rank on the \(i\)-th qubit, then
\[
\prod_{i=1}^n r_i(W)\geq(\dim W)^{n-1}.
\]
The terminology reflects the tensorial nature of the estimate: applying a
covector to one tensor factor is a slice, or contraction, and the inequality
controls these slice ranks for an entire subspace.  It may be viewed as a
quantum-subspace analogue of the Loomis--Whitney projection
inequality~\cite{LoomisWhitney49}, with coordinate projection sizes replaced
by generic contraction ranks.

This inequality implies a local-rank estimate for uniformly random
\(k\)-sets:
\[
\mathbb E_T
\log_2\frac{\dim W}{r_T(W)}
\leq
\frac{k}{n}\log_2\dim W .
\]
Writing \(\dim W=2^{sn}\), this estimate says that a typical \(k\)-local
constraint sees a definite fraction of the current satisfying subspace.  More
precisely, using the arithmetic-geometric mean inequality and Jensen's
inequality, one obtains the logarithmic dimension-drift bound
\[
\mathbb E_T
\log_2\frac{\dim W}{\dim W-r_T(W)}
\geq
\log_2\frac{1}{1-2^{-ks}}.
\]
Thus, when the logarithmic dimension density is \(s\), the next random
constraint decreases it at a rate bounded from below by the expression on the
right-hand side.

The integral defining \(\alpha_\star(k)\) is the continuum form of this drift
calculation.  Starting from the full Hilbert space, where \(s=1\), one
integrates the reciprocal drift down to \(s=0\), obtaining
\[
\alpha_\star(k)
=
\int_0^1
\frac{ds}{
	\log_2\left(\frac{1}{1-2^{-ks}}\right)
} .
\]
After more than \(\alpha_\star(k)n\) constraints, the stopped martingale
argument shows that the satisfying subspace has become trivial with high
probability.

The same viewpoint may also be useful beyond random quantum satisfiability.
A satisfying subspace of a QSAT instance is a code-like subspace cut out by
local constraints, and quantum error-correcting codes are themselves subspaces
of many-qubit Hilbert spaces~\cite{KnillLaflamme97}.  In particular, sparse
local checks are central in the theory of quantum LDPC codes
\cite{TillichZemor14,PanteleevKalachev22}.  The dimension-drift method gives a
way to track how the dimension of such a subspace decays under random local
constraints, and therefore suggests possible applications to nontriviality or
rate estimates for random quantum code spaces.  By itself, the method does not
prove distance or decodability bounds; however, the local contraction-rank
estimates may be relevant to questions of local indistinguishability and
low-weight logical operators.  Conceptually, the method also probes the gap
between product and entangled solution spaces, which is one of the genuinely
quantum features of random QSAT absent from classical random \(k\)-SAT
\cite{LLMSS10,LMRV24}.

This viewpoint differs from the quantum Lov\'asz-local-lemma and quantum
Shearer literature, which gives SAT-side criteria ensuring that the common
kernel remains nonzero under suitable dependency assumptions
\cite{AKS,SMLM16,HLSZ19}.  Here the Shearer-type principle is used in the
opposite direction: it forces a fresh generic constraint to remove a
quantifiable amount of dimension from the current common kernel.  The method
also differs from upper-bound arguments based on finding special local
obstructions.  Instead, it tracks the rank removed by every new random
constraint through a global tensor-projection inequality, which is what leads
to the improvement from order \(2^k\) to order \(2^k/k\).

\section{Generic constraints and contraction ranks}
\label{sec:contraction-ranks}

We now make precise the dimension-drop statement used in the proof strategy.
Let \(W\subseteq\mathcal H_n\) be a nonzero subspace and let \(T\subseteq[n]\).
For a covector \(\varphi_T\in\mathcal H_T^*\), recall the contraction map
\[
C_T(\varphi_T):W\to\mathcal H_{\bar T},
\qquad
C_T(\varphi_T)(\psi)
=
(\varphi_T\otimes I_{\bar T})\psi .
\]
We define the generic contraction rank of \(W\) on \(T\) by
\[
r_T(W)
=
\max_{\varphi_T\in\mathcal H_T^*}
\rank C_T(\varphi_T).
\]
Equivalently, \(r_T(W)\) is the rank of \(C_T(\varphi_T)\) for generic
\(\varphi_T\in\mathcal H_T^*\).

Here and below, ``generic'' is used in the following elementary algebraic
sense.  A condition depending on a covector is generic if it holds outside the
common zero set of finitely many nonzero polynomial equations in the
coordinates of that covector.  Equivalently, it holds on a nonempty
Zariski-open set.  Such exceptional sets have Lebesgue measure zero, and finite
intersections of nonempty Zariski-open sets are again nonempty Zariski-open.

For a singleton \(T=\{i\}\), we write
\[
r_i(W)=r_{\{i\}}(W).
\]

The following lemma is the precise link between a generic local quantum
constraint and the contraction rank.

\begin{lemma}[Generic dimension drop]
	\label{lem:dimension-drop}
	Let \(W\subseteq\mathcal H_n\) be a nonzero subspace, and let \(T\subseteq[n]\).
	If a generic rank-one local constraint with support \(T\) is added, then the
	new subspace
	\[
	W'
	=
	W\cap\ker(\varphi_T\otimes I_{\bar T})
	\]
	satisfies
	\[
	\dim W'=\dim W-r_T(W)
	\]
	almost surely.
\end{lemma}

\begin{proof}
	By definition of \(r_T(W)\), there exists a covector
	\(\eta_T\in\mathcal H_T^*\) such that
	\[
	\rank C_T(\eta_T)=r_T(W).
	\]
	Now regard \(\varphi_T\in\mathcal H_T^*\) as a variable covector.  After
	choosing bases, the condition
	\[
	\rank C_T(\varphi_T)<r_T(W)
	\]
	is equivalent to the vanishing of all \(r_T(W)\times r_T(W)\) minors of the
	matrix of \(C_T(\varphi_T)\).  These minors are polynomial functions of the
	coordinates of \(\varphi_T\).  Since the rank \(r_T(W)\) is attained at
	\(\eta_T\), at least one of these minors is not identically zero.  Therefore
	the exceptional set on which the rank is smaller than \(r_T(W)\) is a proper
	algebraic set, and hence has Lebesgue measure zero.
	
	Since the local covector is sampled from a distribution that is absolutely
	continuous with respect to Lebesgue measure, we have
	\[
	\rank C_T(\varphi_T)=r_T(W)
	\]
	almost surely.  For such a covector,
	\[
	\dim W'
	=
	\dim\ker C_T(\varphi_T)
	=
	\dim W-\rank C_T(\varphi_T)
	=
	\dim W-r_T(W).
	\]
	This proves the claim.
\end{proof}

\section{The multiplicative slice-rank inequality}
\label{sec:slice-rank}

Section~\ref{sec:contraction-ranks} reduced the effect of a generic local
constraint to the contraction rank \(r_T(W)\).  To use this reduction in the
random process, we need deterministic estimates showing that a subspace of
large global dimension must have large contraction ranks on typical local
sets of coordinates.  This section proves such estimates.  The main input is a
multiplicative slice-rank inequality for tensor-product subspaces, which is
then converted into a local-rank shadow bound for uniformly random \(k\)-sets.

\subsection{Coordinate contraction ranks and the multiplicative inequality}
\label{subsec:slice-rank-inequality}

\begin{lemma}[Multiplicative slice-rank inequality]
	\label{lem:multiplicative-slice-rank}
	Let
	\[
	W\subseteq(\mathbb C^2)^{\otimes m}
	\]
	be a nonzero subspace, and set
	\[
	D=\dim W.
	\]
	Then
	\[
	\prod_{i=1}^m r_i(W)\geq D^{m-1}.
	\]
\end{lemma}

\begin{proof}
	We prove the result by induction on \(m\).
	
	For \(m=1\), the claim is
	\[
	r_1(W)\geq D^0=1,
	\]
	which holds because \(W\neq0\), and a generic covector on \(\mathbb C^2\)
	does not annihilate all of \(W\).
	
	Assume the claim known for \((m-1)\)-fold tensor products.  After relabelling
	the coordinates, split off the \(m\)-th tensor factor and write
	\[
	(\mathbb C^2)^{\otimes m}
	=
	\mathbb C^2\otimes\mathcal K,
	\qquad
	\mathcal K=(\mathbb C^2)^{\otimes(m-1)}.
	\]
	Let
	\[
	r=r_m(W).
	\]
	Choose a covector on the first factor whose contraction rank on \(W\) is \(r\).
	After a change of basis in \(\mathbb C^2\), we may assume that this covector
	is \(\langle0|\).  Let
	\[
	\pi_0:W\to\mathcal K
	\]
	be the corresponding \(|0\rangle\)-slice map.  Then
	\[
	\rank \pi_0=r.
	\]
	Set
	\[
	A=\pi_0(W)\subseteq\mathcal K,
	\qquad
	\dim A=r.
	\]
	The kernel of \(\pi_0|_W\) consists of vectors of the form
	\[
	|1\rangle\otimes b,
	\]
	where \(b\) belongs to some subspace
	\[
	B\subseteq\mathcal K.
	\]
	Hence
	\[
	\dim B=D-r.
	\]
	
	The two slice maps of \(W\) onto the \(|0\rangle\)- and
	\(|1\rangle\)-components both have rank at most \(r\), since \(r\) is the
	maximal contraction rank along the \(m\)-th coordinate.  Together, these two
	slices embed \(W\) into \(\mathcal K\oplus\mathcal K\).  Therefore
	\[
	D\leq 2r,
	\]
	and hence
	\[
	r\geq \frac D2.
	\]
	
	Choose a linear section of \(\pi_0:W\to A\).  Then there is a linear map
	\[
	F:A\to\mathcal K
	\]
	such that every element of \(W\) can be written uniquely in the form
	\[
	|0\rangle\otimes a
	+
	|1\rangle\otimes(Fa+b),
	\qquad
	a\in A,\ b\in B.
	\]
	
	Now fix a coordinate \(i<m\), viewed as a coordinate of \(\mathcal K\).  Let
	\[
	r_i(A)
	\qquad\text{and}\qquad
	r_i(B)
	\]
	be the corresponding generic contraction ranks of \(A\) and \(B\).  Choose a
	covector on the \(i\)-th coordinate of \(\mathcal K\) which is generic
	simultaneously for \(A\) and \(B\).  This is possible because the intersection
	of finitely many nonempty Zariski-open sets is again nonempty and Zariski-open.
	
	Let
	\[
	C_i:\mathcal K\to(\mathbb C^2)^{\otimes([m-1]\setminus\{i\})}
	\]
	be the associated contraction.  The contracted image of \(W\) along coordinate
	\(i\) is
	\[
	\left\{
	|0\rangle\otimes C_i a
	+
	|1\rangle\otimes(C_iFa+C_i b)
	:
	a\in A,\ b\in B
	\right\}.
	\]
	Projection onto the \(|0\rangle\)-component has image \(C_i(A)\), of dimension
	\(r_i(A)\).  Its kernel contains the subspace
	\[
	\left\{
	|1\rangle\otimes C_i b:b\in B
	\right\},
	\]
	which has dimension \(r_i(B)\).  Hence this contracted image has dimension at
	least
	\[
	r_i(A)+r_i(B).
	\]
	Since \(r_i(W)\) is the maximal contraction rank along coordinate \(i\), we
	obtain
	\[
	r_i(W)\geq r_i(A)+r_i(B)
	\qquad\text{for every } i<m.
	\]
	
	By the induction hypothesis applied to \(A\subseteq\mathcal K\),
	\[
	\prod_{i=1}^{m-1} r_i(A)\geq r^{m-2}.
	\]
	If \(B=0\), then \(D=r\), and therefore
	\[
	\prod_{i=1}^m r_i(W)
	\geq
	r\prod_{i=1}^{m-1}r_i(A)
	\geq
	D\cdot D^{m-2}
	=
	D^{m-1}.
	\]
	Thus we may assume \(B\neq0\).  Applying the induction hypothesis to
	\(B\subseteq\mathcal K\) gives
	\[
	\prod_{i=1}^{m-1} r_i(B)\geq (D-r)^{m-2}.
	\]
	
	Using \(r_i(W)\geq r_i(A)+r_i(B)\), we get
	\[
	\prod_{i=1}^{m-1} r_i(W)
	\geq
	\prod_{i=1}^{m-1}\bigl(r_i(A)+r_i(B)\bigr).
	\]
	We shall use the multiplicative Minkowski inequality
	\[
	\left(
	\prod_{i=1}^{m-1}(x_i+y_i)
	\right)^{1/(m-1)}
	\geq
	\left(
	\prod_{i=1}^{m-1}x_i
	\right)^{1/(m-1)}
	+
	\left(
	\prod_{i=1}^{m-1}y_i
	\right)^{1/(m-1)}
	\]
	for nonnegative \(x_i,y_i\).  This follows from the concavity and homogeneity
	of the geometric mean.  Therefore
	\[
	\prod_{i=1}^{m-1} r_i(W)
	\geq
	\left(
	r^{(m-2)/(m-1)}
	+
	(D-r)^{(m-2)/(m-1)}
	\right)^{m-1}.
	\]
	Multiplying by \(r_m(W)=r\), we obtain
	\[
	\prod_{i=1}^m r_i(W)
	\geq
	r
	\left(
	r^{(m-2)/(m-1)}
	+
	(D-r)^{(m-2)/(m-1)}
	\right)^{m-1}.
	\]
	
	It remains to check that the right-hand side is at least \(D^{m-1}\).  Set
	\[
	x=\frac rD.
	\]
	Since \(r\geq D/2\), we have \(x\in[1/2,1]\).  Dividing by \(D^{m-1}\), it
	suffices to prove
	\[
	x
	\left(
	x^{(m-2)/(m-1)}
	+
	(1-x)^{(m-2)/(m-1)}
	\right)^{m-1}
	\geq 1.
	\]
	Let
	\[
	a=x^{1/(m-1)},
	\qquad
	b=(1-x)^{1/(m-1)}.
	\]
	Then \(a\geq b\) and
	\[
	a^{m-1}+b^{m-1}=1.
	\]
	The left-hand side above, raised to the power \(1/(m-1)\), is
	\[
	a(a^{m-2}+b^{m-2})
	=
	a^{m-1}+ab^{m-2}.
	\]
	Since \(a\geq b\),
	\[
	ab^{m-2}\geq b^{m-1}.
	\]
	Thus
	\[
	a^{m-1}+ab^{m-2}
	\geq
	a^{m-1}+b^{m-1}
	=
	1.
	\]
	This proves the scalar inequality and completes the induction.
\end{proof}

\subsection{Local-rank shadow inequalities}
\label{sec:shadow}

The multiplicative inequality controls the product of the one-coordinate
contraction ranks.  We now convert it into an averaged estimate, first for a
uniformly random coordinate and then for a uniformly random \(k\)-set.  The
second form is the one used in the dimension-drift argument of
Section~\ref{sec:dimension-drift-proof}.

\begin{corollary}[Single-site rank shadow]
	\label{cor:single-site-shadow}
	Let
	\[
	W\subseteq(\mathbb C^2)^{\otimes n}
	\]
	be nonzero with \(\dim W=D\).  Then
	\[
	\frac1n\sum_{i=1}^n
	\log_2\frac{D}{r_i(W)}
	\leq
	\frac1n\log_2D.
	\]
	Equivalently, if \(I\) is uniformly random in \([n]\), then
	\[
	\mathbb E_I
	\log_2\frac{D}{r_I(W)}
	\leq
	\frac1n\log_2D.
	\]
\end{corollary}

\begin{proof}
	Taking logarithms in Lemma~\ref{lem:multiplicative-slice-rank} gives
	\[
	\sum_{i=1}^n\log_2 r_i(W)
	\geq
	(n-1)\log_2D.
	\]
	Equivalently,
	\[
	\sum_{i=1}^n
	\log_2\frac{D}{r_i(W)}
	\leq
	\log_2D.
	\]
	Dividing by \(n\) proves the claim.
\end{proof}

\begin{corollary}[Local-rank shadow for \(k\) coordinates]
	\label{cor:k-site-shadow}
	Let
	\[
	W\subseteq(\mathbb C^2)^{\otimes n}
	\]
	be nonzero with \(\dim W=D\).  Let \(T\subseteq[n]\) be uniformly random among
	all \(k\)-subsets.  Then
	\[
	\mathbb E_T
	\log_2\frac{D}{r_T(W)}
	\leq
	\frac{k}{n}\log_2D.
	\]
\end{corollary}

\begin{proof}
	Choose an ordered \(k\)-tuple of distinct coordinates uniformly at random.
	Starting from
	\[
	W_0=W,
	\]
	perform generic one-site contractions successively along these coordinates.
	Let
	\[
	D_j=\dim W_j
	\]
	be the image dimension after \(j\) contractions.  Thus \(D_0=D\).
	
	Conditioned on \(W_j\), the next coordinate is uniformly distributed among
	the remaining \(n-j\) tensor factors.  Applying
	Corollary~\ref{cor:single-site-shadow} to \(W_j\) on these remaining factors
	gives
	\[
	\mathbb E
	\left[
	\log_2\frac{D_j}{D_{j+1}}
	\,\middle|\, W_j
	\right]
	\leq
	\frac1{n-j}\log_2D_j.
	\]
	Equivalently,
	\[
	\mathbb E[\log_2D_{j+1}\mid W_j]
	\geq
	\left(1-\frac1{n-j}\right)\log_2D_j.
	\]
	Iterating this estimate for \(j=0,\ldots,k-1\) yields
	\[
	\mathbb E\log_2D_k
	\geq
	\frac{n-k}{n}\log_2D.
	\]
	Consequently,
	\[
	\mathbb E
	\log_2\frac{D}{D_k}
	\leq
	\frac{k}{n}\log_2D.
	\]
	
	For a fixed unordered set \(T\), the generic rank \(r_T(W)\) over
	\(\mathcal H_T^*\) is the maximum possible rank of \(C_T(\varphi_T)\).  Hence
	it is at least the rank obtained from any product covector on the coordinates
	of \(T\).  Applying this to the generic product covectors used above gives
	\[
	r_T(W)\geq D_k.
	\]
	Therefore
	\[
	\log_2\frac{D}{r_T(W)}
	\leq
	\log_2\frac{D}{D_k}.
	\]
	Taking expectations over \(T\) gives
	\[
	\mathbb E_T
	\log_2\frac{D}{r_T(W)}
	\leq
	\frac{k}{n}\log_2D,
	\]
	as claimed.
\end{proof}

\section{Dimension drift and proof of the main theorem}\label{sec:dimension-drift-proof}

Section~\ref{sec:slice-rank} gave a deterministic estimate on the contraction
rank seen by a uniformly random \(k\)-set.  We now turn this static estimate
into a dynamical statement about the random quantum \(k\)-SAT process.  At a
given time, the current satisfying subspace has some logarithmic dimension
density \(s\).  The local-rank shadow bound implies that the next random
generic constraint produces a definite expected logarithmic drop, depending
only on \(s\).  Integrating this drift, with a truncation to obtain bounded
increments, proves that after more than \(\alpha_\star(k)n\) constraints the
satisfying subspace is trivial with high probability.

\subsection{Expected logarithmic dimension drift}\label{sec:drift}
	
	Let \(W\subseteq\mathcal H_n\) be nonzero with
	\[
	D=\dim W=2^{sn}.
	\]
	For a uniformly random \(k\)-subset \(T\subseteq[n]\), define
	\[
	q_T(W)=\frac{r_T(W)}{D}.
	\]
	The local-rank shadow inequality gives
	\[
	\mathbb E_T\log_2\frac1{q_T(W)}
	\le
	ks.
	\]
	By the genericity convention above, a generic quantum \(k\)-SAT
	constraint on \(T\) changes the dimension from
	\[
	D
	\]
	to
	\[
	D-r_T(W)=D(1-q_T(W))
	\]
	almost surely. Thus, with the convention that the logarithmic drop is
	\(+\infty\) if \(r_T(W)=D\), the logarithmic dimension drop is
	\[
	X_T(W)
	=
	\log_2\frac{D}{D-r_T(W)}
	=
	\log_2\frac1{1-q_T(W)}.
	\]
	
	Set
	\[
	h_T(W)=\log_2\frac1{q_T(W)}.
	\]
	Then
	\[
	X_T(W)=g(h_T(W)),
	\]
	where
	\[
	g(h)=\log_2\frac1{1-2^{-h}},\qquad h>0,
	\]
	and \(g(0)=+\infty\) by convention. The function \(g\) is decreasing and
	convex on \((0,\infty)\). Therefore,
	by Jensen's inequality,
	\[
	\mathbb E_T X_T(W)
	=
	\mathbb E_T g(h_T(W))
	\ge
	g\bigl(\mathbb E_T h_T(W)\bigr).
	\]
	Since \(g\) is decreasing and
	\[
	\mathbb E_T h_T(W)\le ks,
	\]
	we obtain
	\[
	\mathbb E_T X_T(W)
	\ge
	g(ks)
	=
	\log_2\frac1{1-2^{-ks}}.
	\]
	
\subsection{Stopped martingale integration}\label{sec:martingale}
	
	For \(C>0\), define
	\[
	g_C(h)=\min\{g(h),C\},
	\qquad h\in(0,k],
	\]
	and set \(g_C(0)=C\). Thus \(g_C\) is finite and positive on \([0,k]\).
	Let \(\phi_C\) be the largest convex nonincreasing function on \([0,k]\)
	such that
	\[
	\phi_C(h)\le g_C(h)
	\qquad\text{for all }h\in[0,k].
	\]
	Define
	\[
	\mu_{k,C}(s)=\phi_C(ks),
	\qquad s\in[0,1],
	\]
	and
	\[
	\alpha_{k,C}
	=
	\int_0^1\frac{ds}{\mu_{k,C}(s)}.
	\]
	
	\begin{lemma}[Truncated drift]
		Let \(W\subseteq\mathcal H_n\) be nonzero with
		\[
		\dim W=2^{sn}.
		\]
		Then
		\[
		\mathbb E_T
		\min\left\{
		\log_2\frac{\dim W}{\dim W-r_T(W)},C
		\right\}
		\ge
		\mu_{k,C}(s),
		\]
		where the expression inside the minimum is interpreted as \(+\infty\) if
		\(r_T(W)=\dim W\).
	\end{lemma}
	
	\begin{proof}
		Fix \(T\). We first record the elementary lower bound
		\[
		r_T(W)\ge 2^{-k}\dim W.
		\]
		Choose a product basis \((e_a)_{a=1}^{2^k}\) of \(\mathcal H_T\), with dual
		basis \((e_a^*)_{a=1}^{2^k}\). The slice maps
		\[
		(e_a^*\otimes I_{\bar T})|_W:W\to\mathcal H_{\bar T}
		\]
		jointly inject \(W\) into \(\mathcal H_{\bar T}^{\oplus 2^k}\). Each of these
		maps has rank at most \(r_T(W)\), because \(r_T(W)\) is the maximal
		contraction rank over all covectors in \(\mathcal H_T^*\). Therefore
		\[
		\dim W\le \sum_{a=1}^{2^k}\rank (e_a^*\otimes I_{\bar T})|_W
		\le 2^k r_T(W),
		\]
		which proves the claimed lower bound.
		
		Consequently
		\[
		q_T(W)=\frac{r_T(W)}{\dim W}\ge 2^{-k},
		\]
		and hence
		\[
		h_T(W)=\log_2\frac1{q_T(W)}\in[0,k].
		\]
		The truncated logarithmic drop is
		\[
		g_C(h_T(W)).
		\]
		Since
		\[
		\phi_C\le g_C,
		\]
		we have
		\[
		\mathbb E_T g_C(h_T(W))
		\ge
		\mathbb E_T\phi_C(h_T(W)).
		\]
		The function \(\phi_C\) is convex, so Jensen's inequality gives
		\[
		\mathbb E_T\phi_C(h_T(W))
		\ge
		\phi_C\bigl(\mathbb E_T h_T(W)\bigr).
		\]
		Since \(\phi_C\) is nonincreasing and
		\[
		\mathbb E_T h_T(W)\le ks,
		\]
		we get
		\[
		\phi_C\bigl(\mathbb E_T h_T(W)\bigr)
		\ge
		\phi_C(ks)
		=
		\mu_{k,C}(s).
		\]
	\end{proof}
	
	\begin{proposition}[High-probability hitting bound]
		Fix \(k\ge3\), \(C>0\), and \(\varepsilon>0\). In the model where
		\(k\)-subsets are sampled independently and uniformly from
		\(\binom{[n]}k\), a random generic quantum \(k\)-SAT instance with
		\[
		m\ge(\alpha_{k,C}+\varepsilon)n
		\]
		constraints is unsatisfiable with high probability.
	\end{proposition}
	
	\begin{proof}
		Let \(W_t\) be the satisfying subspace after the first \(t\) constraints,
		and let
		\[
		D_t=\dim W_t.
		\]
		Let \(\mathcal F_t\) be the sigma-field generated by the first \(t\)
		chosen \(k\)-sets and their local covectors. While \(D_t>0\), define
		\[
		L_t=\log_2D_t,
		\qquad
		s_t=\frac{L_t}{n}.
		\]
		For the next constraint define the truncated drop by
		\[
		Y_{t+1}
		=
		\begin{cases}
		\min\{L_t-L_{t+1},C\}, & D_t>0\text{ and }D_{t+1}>0,\\[2mm]
		C, & D_t>0\text{ and }D_{t+1}=0,\\[2mm]
		0, & D_t=0.
		\end{cases}
		\]
		Thus \(0\le Y_{t+1}\le C\). Conditional on \(\mathcal F_t\) and on
		\(D_t>0\), the next set \(T_{t+1}\) is uniform and independent of
		\(\mathcal F_t\), and the next local covector is generic for the fixed
		subspace \(W_t\) almost surely. Hence the truncated drift lemma gives
		\[
		\mathbb E[Y_{t+1}\mid\mathcal F_t]
		\ge
		\mu_{k,C}(s_t).
		\]
		
		Choose \(\eta>0\) small enough that
		\[
		(1+\eta)\left(\alpha_{k,C}+\frac{\varepsilon}{4}\right)
		\le
		\alpha_{k,C}+\frac{\varepsilon}{2}.
		\]
		Since \(\mu_{k,C}\) is positive and continuous on \([0,1]\), the function
		\(1/\mu_{k,C}\) is Riemann integrable. Choose a finite partition
		\[
		1=s_0>s_1>\cdots>s_M=0
		\]
		fine enough that the corresponding upper Riemann sum satisfies
		\[
		\sum_{j=1}^M
		\frac{s_{j-1}-s_j}{\mu_{k,C}(s_{j-1})}
		\le
		\alpha_{k,C}+\frac{\varepsilon}{4}.
		\]
		Set
		\[
		\Delta_j=s_{j-1}-s_j,
		\qquad
		m_j=\mu_{k,C}(s_{j-1}).
		\]
		Since \(\phi_C\le g_C\le C\), we have \(0<m_j\le C\).
		
		Stage \(j\) begins once either the process has already been killed, or else
		\(s_t\le s_{j-1}\). It ends once either \(D_t=0\), or else
		\(s_t\le s_j\). During stage \(j\), before it has ended, we have
		\[
		s_t\in(s_j,s_{j-1}],
		\]
		and since \(\mu_{k,C}\) is nonincreasing,
		\[
		\mu_{k,C}(s_t)\ge \mu_{k,C}(s_{j-1})=m_j.
		\]
		Expose
		\[
		N_j=
		\left\lceil
		(1+\eta)\frac{\Delta_j n}{m_j}
		\right\rceil
		\]
		fresh constraints during stage \(j\).
		
		We now make the stopping argument explicit. During these \(N_j\) trials,
		define modified increments \(\widetilde Y_1^{(j)},\ldots,
		\widetilde Y_{N_j}^{(j)}\) as follows. Before stage \(j\) has ended,
		\(\widetilde Y_\ell^{(j)}\) is the actual truncated drop of the next
		constraint; after stage \(j\) has ended, set
		\(\widetilde Y_\ell^{(j)}=m_j\). Then
		\[
		0\le \widetilde Y_\ell^{(j)}\le C
		\]
		and, with respect to the natural stopped filtration during this stage,
		\[
		\mathbb E[\widetilde Y_\ell^{(j)}\mid\mathcal G_{\ell-1}^{(j)}]
		\ge m_j.
		\]
		Therefore the variables
		\[
		Z_\ell^{(j)}=
		\widetilde Y_\ell^{(j)}-
		\mathbb E[\widetilde Y_\ell^{(j)}\mid\mathcal G_{\ell-1}^{(j)}]
		\]
		form a martingale-difference sequence with \(|Z_\ell^{(j)}|\le C\). Since
		\[
		N_jm_j\ge (1+\eta)\Delta_j n,
		\]
		Azuma-Hoeffding gives
		\[
		\mathbb P
		\left(
		\sum_{\ell=1}^{N_j}\widetilde Y_\ell^{(j)}
		<
		\Delta_j n
		\right)
		\le
		\exp(-c_j n)
		\]
		for some constant \(c_j>0\) depending only on \(C,\eta,\Delta_j,m_j\), but
		not on \(n\).
		
		If stage \(j\) fails to end within these \(N_j\) constraints, then all
		modified increments are actual truncated drops. The actual logarithmic
		decrease during the stage is then less than \(\Delta_j n\), and the sum of
		the truncated drops is at most the actual logarithmic decrease. Hence the
		event
		\[
		\left\{\sum_{\ell=1}^{N_j}\widetilde Y_\ell^{(j)}<\Delta_j n\right\}
		\]
		occurs. Thus stage \(j\) fails with probability at most \(\exp(-c_j n)\).
		A union bound over the finitely many stages shows that all stages succeed
		with probability \(1-o(1)\).
		
		The number of constraints used by the stages is at most
		\[
		\sum_{j=1}^M N_j
		\le
		(1+\eta)n
		\sum_{j=1}^M
		\frac{\Delta_j}{\mu_{k,C}(s_{j-1})}
		+O(1).
		\]
		By the choice of the partition and of \(\eta\), this is at most
		\[
		\left(\alpha_{k,C}+\frac{\varepsilon}{2}\right)n+O(1).
		\]
		For all sufficiently large \(n\), this is at most
		\[
		\left(\alpha_{k,C}+\frac{3\varepsilon}{4}\right)n.
		\]
		If all stages succeed, then after these constraints either \(D_t=0\), or
		else \(s_t\le0\), hence \(D_t\le1\). If \(D_t=1\), then one further generic
		quantum \(k\)-SAT constraint kills the remaining one-dimensional subspace
		with probability one: for any nonzero vector \(\psi\in\mathcal H_n\) and
		any fixed \(T\), the map \(\varphi_T\mapsto(\varphi_T\otimes I_{\bar T})\psi\)
		is not identically zero. The set of covectors annihilating \(\psi\) is
		therefore a proper linear subspace of \(\mathcal H_T^*\).
		
		This extra constraint is covered by the remaining linear slack: for large
		\(n\),
		\[
		\left(\alpha_{k,C}+\frac{3\varepsilon}{4}\right)n+1
		\le
		(\alpha_{k,C}+\varepsilon)n.
		\]
		Thus any instance with at least \((\alpha_{k,C}+\varepsilon)n\) constraints
		is unsatisfiable with probability \(1-o(1)\).
	\end{proof}

\subsection{Completion of the proof}\label{sec:completion}

\subsubsection{Removing the truncation}
	
	\begin{lemma}[Limit of the truncated time]
		As \(C\to\infty\),
		\[
		\alpha_{k,C}
		\downarrow
		\alpha_\star(k),
		\]
		where
		\[
		\alpha_\star(k)
		=
		\int_0^1
		\frac{ds}{
			\log_2\left(\frac1{1-2^{-ks}}\right)
		}.
		\]
	\end{lemma}
	
	\begin{proof}
		If \(C_1\le C_2\), then \(g_{C_1}\le g_{C_2}\). Hence the admissible
		class of convex nonincreasing minorants for \(C_1\) is contained in the
		corresponding class for \(C_2\). Therefore \(\phi_{C_1}\le\phi_{C_2}\),
		\(\mu_{k,C_1}\le\mu_{k,C_2}\), and \(\alpha_{k,C}\) is nonincreasing in
		\(C\).
		
		For \(h>0\),
		\[
		g_C(h)\uparrow g(h)
		\qquad (C\to\infty).
		\]
		We claim that \(\phi_C(h)\to g(h)\) for every \(h\in(0,k]\). Fix
		\(\delta\in(0,k]\), and let \(\ell_\delta\) be the tangent line to \(g\) at
		\(\delta\). By convexity,
		\[
		\ell_\delta(h)\le g(h)
		\qquad (h>0).
		\]
		Define
		\[
		\psi_\delta(h)=
		\begin{cases}
		\ell_\delta(h), & 0\le h\le\delta,\\
		g(h), & \delta\le h\le k.
		\end{cases}
		\]
		This function is convex and nonincreasing on \([0,k]\). Since
		\(\ell_\delta\) is finite on \([0,\delta]\), for all sufficiently large
		\(C\) we have \(\psi_\delta\le g_C\) on \([0,k]\): on \([\delta,k]\) the
		cap is inactive, while on \([0,\delta]\) we choose \(C\) larger than
		\(\max_{[0,\delta]}\ell_\delta\). Hence \(\psi_\delta\) is an admissible
		minorant for all sufficiently large \(C\). Since also \(\phi_C\le g_C=g\)
		on \([\delta,k]\), we obtain
		\[
		\phi_C(h)=g(h)
		\qquad (h\in[\delta,k])
		\]
		for all sufficiently large \(C\). Because \(\delta>0\) is arbitrary, this
		proves \(\phi_C(h)\to g(h)\) for every \(h>0\).
		
		Therefore
		\[
		\mu_{k,C}(s)=\phi_C(ks)
		\to
		g(ks)
		\]
		for every \(s>0\). For all sufficiently large \(C\), namely \(C\ge g(k)\),
		the constant function \(g(k)\) is a convex nonincreasing minorant of
		\(g_C\). Thus
		\[
		\mu_{k,C}(s)\ge g(k)>0
		\]
		for all \(s\in[0,1]\) and all sufficiently large \(C\). Hence
		\[
		\frac1{\mu_{k,C}(s)}
		\le
		\frac1{g(k)}.
		\]
		Dominated convergence gives
		\[
		\alpha_{k,C}
		=
		\int_0^1\frac{ds}{\mu_{k,C}(s)}
		\to
		\int_0^1
		\frac{ds}{g(ks)}
		=
		\alpha_\star(k).
		\]
		Together with monotonicity, this proves \(\alpha_{k,C}\downarrow
		\alpha_\star(k)\).
	\end{proof}

\subsubsection{Transfer to the simple random hypergraph model}
	
	The proof above used the model where \(k\)-subsets are sampled
	independently with replacement. For fixed \(k\ge3\) and \(m=O(n)\), the
	probability of a repeated \(k\)-set is at most
	\[
	\binom m2 \binom nk^{-1}
	=
	O(n^{2-k})
	=
	o(1).
	\]
	Conditioned on no repetitions, the sampled set of constraints has the
	usual simple random \(k\)-uniform hypergraph distribution. Hence the same
	upper bound holds in the simple random quantum \(k\)-SAT model.
	
\subsubsection{Large-\(k\) asymptotics}
	
	With the change of variables
	\[
	u=ks,
	\]
	we have
	\[
	\alpha_\star(k)
	=
	\frac1k
	\int_0^k
	\frac{du}{
		\log_2\left(\frac1{1-2^{-u}}\right)
	}.
	\]
	As \(u\to\infty\),
	\[
	\log_2\left(\frac1{1-2^{-u}}\right)
	\sim
	\frac{2^{-u}}{\ln2}.
	\]
	Therefore
	\[
	\frac1{
		\log_2\left(\frac1{1-2^{-u}}\right)
	}
	\sim
	(\ln2)2^u.
	\]
	To justify the endpoint asymptotic, fix \(\rho>0\). For all sufficiently
	large \(U\), the integrand lies between
	\((1-\rho)(\ln2)2^u\) and \((1+\rho)(\ln2)2^u\) for \(u\ge U\). The
	contribution of \([0,U]\) is \(O_U(1)\), which is negligible compared with
	\(2^k\). Hence
	\[
	\int_0^k
	\frac{du}{
		\log_2\left(\frac1{1-2^{-u}}\right)
	}
	\sim
	\int_0^k (\ln2)2^u\,du
	=
	2^k-1.
	\]
	Consequently
	\[
	\alpha_\star(k)
	\sim
	\frac{2^k}{k}
	\qquad (k\to\infty).
	\]

	\begin{remark}
		For \(k=3\),
		\[
		\alpha_\star(3)
		=
		\int_0^1
		\frac{ds}{
			\log_2\left(\frac1{1-2^{-3s}}\right)
		}
		\approx 1.947.
		\]
		This value is obtained by direct one-dimensional numerical integration; a
		short Python script for evaluating \(\alpha_\star(k)\) is available at
		\[
		\text{\url{https://github.com/rjeanbernoulli/qksat-alpha-star}}.
		\]
	\end{remark}

\subsubsection{Proof of the main theorem}

\begin{proof}[Proof of Theorem~\ref{thm:main}]
Choose \(C\) large enough that
\[
\alpha_{k,C}\le \alpha_\star(k)+\frac{\varepsilon}{2}.
\]
The high-probability hitting proposition, applied with \(\varepsilon/2\) in place of \(\varepsilon\), shows that the with-replacement model is unsatisfiable with high probability once
\[
 m\ge (\alpha_\star(k)+\varepsilon)n.
\]
The transfer argument above gives the same conclusion in the simple random \(k\)-uniform hypergraph model. The asymptotic statement \(\alpha_\star(k)\sim 2^k/k\) was proved in the preceding subsection. This proves all claims of the theorem.
\end{proof}

\bigskip

\noindent\textsc{Jean Bernoulli Ravelomanana}

\noindent\emph{Email address:}
\href{mailto:rjeanbernoulli@gmail.com}{\texttt{rjeanbernoulli@gmail.com}}


\begin{thebibliography}{99}

\bibitem{Bravyi06}
S.~Bravyi,
\emph{Efficient algorithm for a quantum analogue of 2-SAT},
arXiv:quant-ph/0602108, 2006.

\bibitem{LMS10}
C.~R. Laumann, R.~Moessner, A.~Scardicchio, and S.~L. Sondhi,
\emph{Phase transitions and random quantum satisfiability},
Quantum Information \& Computation \textbf{10} (2010), no.~1--2, 0001--0015;
arXiv:0903.1904.

\bibitem{LLMSS10}
C.~R. Laumann, A.~M. L\"auchli, R.~Moessner, A.~Scardicchio, and S.~L. Sondhi,
\emph{Product, generic, and random generic quantum satisfiability},
Physical Review A \textbf{81} (2010), 062345;
arXiv:0910.2058.

\bibitem{BMR}
S.~Bravyi, C.~Moore, and A.~Russell,
\emph{Bounds on the quantum satisfiability threshold},
in \emph{Proceedings of Innovations in Computer Science (ICS)}, 2010,
pp.~482--489;
arXiv:0907.1297.

\bibitem{AKS}
A.~Ambainis, J.~Kempe, and O.~Sattath,
\emph{A quantum Lov\'asz local lemma},
Journal of the ACM \textbf{59} (2012), no.~5, Article~24.

\bibitem{SMLM16}
O.~Sattath, S.~C. Morampudi, C.~R. Laumann, and R.~Moessner,
\emph{When a local Hamiltonian must be frustration-free},
Proceedings of the National Academy of Sciences \textbf{113} (2016), no.~23,
6433--6437;
arXiv:1509.07766.

\bibitem{HLSZ19}
K.~He, Q.~Li, X.~Sun, and J.~Zhang,
\emph{Quantum Lov\'asz local lemma: Shearer's bound is tight},
in \emph{Proceedings of the 51st Annual ACM Symposium on Theory of Computing
(STOC)}, 2019, pp.~461--472;
arXiv:1804.07055.

\bibitem{MorampudiHSLM17}
S.~C. Morampudi, B.~Hsu, S.~L. Sondhi, R.~Moessner, and C.~R. Laumann,
\emph{Clustering in Hilbert space of a quantum optimization problem},
Physical Review A \textbf{96} (2017), 042303;
arXiv:1704.00238.

\bibitem{LoomisWhitney49}
L.~H. Loomis and H.~Whitney,
\emph{An inequality related to the isoperimetric inequality},
Bulletin of the American Mathematical Society \textbf{55} (1949), no.~10,
961--962.

\bibitem{MezardMontanari09}
M. M{\'e}zard and A. Montanari,
\emph{Information, Physics, and Computation},
Oxford University Press, Oxford, 2009.

\bibitem{Friedgut99}
E. Friedgut,
Sharp thresholds of graph properties, and the \(k\)-SAT problem,
\emph{Journal of the American Mathematical Society} \textbf{12} (1999),
no.~4, 1017--1054.

\bibitem{AchlioptasPeres04}
D. Achlioptas and Y. Peres,
The threshold for random \(k\)-SAT is \(2^k \log 2 - O(k)\),
\emph{Journal of the American Mathematical Society} \textbf{17} (2004),
no.~4, 947--973.

\bibitem{DingSlySun22}
J. Ding, A. Sly and N. Sun,
Proof of the satisfiability conjecture for large \(k\),
\emph{Annals of Mathematics} \textbf{196} (2022),
no.~1, 1--388.

\bibitem{AchlioptasMoore06}
D. Achlioptas and C. Moore,
Random \(k\)-SAT: two moments suffice to cross a sharp threshold,
\emph{SIAM Journal on Computing} \textbf{36} (2006),
no.~3, 740--762.

\bibitem{MezardParisiZecchina02}
M. M{\'e}zard, G. Parisi and R. Zecchina,
Analytic and algorithmic solution of random satisfiability problems,
\emph{Science} \textbf{297} (2002),
no.~5582, 812--815.

\bibitem{MertensMezardZecchina06}
S. Mertens, M. M{\'e}zard and R. Zecchina,
Threshold values of random \(K\)-SAT from the cavity method,
\emph{Random Structures \& Algorithms} \textbf{28} (2006),
no.~3, 340--373.

\bibitem{KrzakalaMontanariRicciSemerjianZdeborova07}
F. Krzakala, A. Montanari, F. Ricci-Tersenghi, G. Semerjian and L. Zdeborov{\'a},
Gibbs states and the set of solutions of random constraint satisfaction problems,
\emph{Proceedings of the National Academy of Sciences} \textbf{104} (2007),
no.~25, 10318--10323.

\bibitem{AchlioptasCojaOghlan08}
D. Achlioptas and A. Coja-Oghlan,
Algorithmic barriers from phase transitions,
in \emph{Proceedings of the 49th Annual IEEE Symposium on Foundations of Computer Science},
FOCS 2008, pp.~793--802.


\bibitem{GossetNagaj16}
D. Gosset and D. Nagaj,
Quantum \(3\)-SAT is \(\mathrm{QMA}_1\)-complete,
\emph{SIAM Journal on Computing} \textbf{45} (2016),
no.~3, 1080--1128. 

\bibitem{LMRV24}
J. Lee, N. Macris, J. B. Ravelomanana and P. Vantalon,
The PRODSAT phase of random quantum satisfiability,
arXiv:2404.18447, 2024.

\bibitem{KirkpatrickSelman94}
S. Kirkpatrick and B. Selman,
Critical behavior in the satisfiability of random Boolean expressions,
\emph{Science} \textbf{264} (1994), no.~5163, 1297--1301.

\bibitem{MonassonZecchina99}
R. Monasson, R. Zecchina, S. Kirkpatrick, B. Selman and L. Troyansky,
Determining computational complexity from characte\bibitem{KirkpatrickSelman94}
S. Kirkpatrick and B. Selman,
Critical behavior in the satisfiability of random Boolean expressions,
\emph{Science} \textbf{264} (1994), no.~5163, 1297--1301.

\bibitem{MonassonZecchina99}
R. Monasson, R. Zecchina, S. Kirkpatrick, B. Selman and L. Troyansky,
Determining computational complexity from characteristic phase transitions,
\emph{Nature} \textbf{400} (1999), 133--137.

\bibitem{KnillLaflamme97}
E. Knill and R. Laflamme,
Theory of quantum error-correcting codes,
\emph{Physical Review A} \textbf{55} (1997), 900--911.

\bibitem{TillichZemor14}
J.-P. Tillich and G. Zemor,
Quantum LDPC codes with positive rate and minimum distance proportional to
\(\sqrt n\),
\emph{IEEE Transactions on Information Theory} \textbf{60} (2014),
no.~2, 1193--1202.

\bibitem{PanteleevKalachev22}
P. Panteleev and G. Kalachev,
Asymptotically good quantum and locally testable classical LDPC codes,
in \emph{Proceedings of the 54th Annual ACM SIGACT Symposium on Theory of
Computing}, STOC 2022, pp.~375--388.

\end{thebibliography}
\end{document}